\pdfoutput=1
\RequirePackage{ifpdf}
\ifpdf 
\documentclass[pdftex]{sigma}
\else
\documentclass{sigma}
\fi

\numberwithin{equation}{section}

\newtheorem{Theorem}{Theorem}[section]
\newtheorem*{Theorem*}{Theorem}
\newtheorem{Corollary}[Theorem]{Corollary}
\newtheorem{Lemma}[Theorem]{Lemma}
\newtheorem{Proposition}[Theorem]{Proposition}
 { \theoremstyle{definition}
\newtheorem{Definition}[Theorem]{Definition}

\newtheorem{Remark}[Theorem]{Remark} }

\usepackage{mathrsfs,tikz,eucal}

\begin{document}
\allowdisplaybreaks

\renewcommand{\thefootnote}{}

\newcommand{\arXivNumber}{2306.16323}

\renewcommand{\PaperNumber}{100}

\FirstPageHeading

\ShortArticleName{Jacobi Beta Ensemble and $b$-Hurwitz Numbers}

\ArticleName{Jacobi Beta Ensemble and $\boldsymbol{b}$-Hurwitz Numbers\footnote{This paper is a~contribution to the Special Issue on Evolution Equations, Exactly Solvable Models and Random Matrices in honor of Alexander Its' 70th birthday. The~full collection is available at \href{https://www.emis.de/journals/SIGMA/Its.html}{https://www.emis.de/journals/SIGMA/Its.html}}}

\Author{Giulio RUZZA~$^{\rm ab}$}

\AuthorNameForHeading{G.~Ruzza}

\Address{$^{\rm a)}$~Grupo de F\'{i}sica Matem\'{a}tica, Campo Grande Edif\'{i}cio C6, 1749-016, Lisboa, Portugal}
\Address{$^{\rm b)}$~Departamento de Matem\'atica, Faculdade de Ci\^encias da Universidade de Lisboa,\\
\hphantom{$^{\rm b)}$}~Campo Grande Edif\'{i}cio C6, 1749-016, Lisboa, Portugal}
\EmailD{\href{mailto:gruzza@fc.ul.pt}{gruzza@fc.ul.pt}}

\ArticleDates{Received July 03, 2023, in final form November 29, 2023; Published online December 19, 2023}

\Abstract{We express correlators of the Jacobi $\beta$ ensemble in terms of (a special case of) $b$-Hurwitz numbers, a deformation of Hurwitz numbers recently introduced by Chapuy and Do\l\c{e}ga. The proof relies on Kadell's generalization of the Selberg integral. The Laguerre limit is also considered. All the relevant $b$-Hurwitz numbers are interpreted (following Bonzom, Chapuy, and Do\l\c{e}ga) in terms of colored monotone Hurwitz maps.}


\Classification{15B52; 05E05; 05E16}

\renewcommand{\thefootnote}{\arabic{footnote}}
\setcounter{footnote}{0}

\section{Introduction and statement of results}

Hurwitz numbers count branched coverings of the sphere by a Riemann surface with prescribed ramification profiles.
Hurwitz himself~\cite{Hurwitz} showed that this geometric counting problem boils down, via monodromy representation, to a combinatorial one.
The latter is the problem of counting factorizations of the identity in the symmetric group with factors in prescribed conjugacy classes.
Today, Hurwitz numbers have been generalized in various directions and are the subject of renewed interest because of their connections to integrable systems~\cite{HarnadPaquet,HarnadOrlov,Okounkov} and enumerative geometry~\cite{Dijkgraaf,ELSV,OkounkovPandharipande}.

There are many matrix models connected with (various versions of) Hurwitz numbers, e.g., the Harish-Chandra--Itzykson--Zuber integral~\cite{GGPN} and the Br\'ezin--Gross--Witten model~\cite{Novak}, as well as externally coupled Br\'ezin--Hikami type models with a Meijer-G weight~\cite{BertolaHarnad}.
A matrix model for simple Hurwitz numbers was given in~\cite{BEMS}.
Moreover, it has been shown~\cite{GisonniGravaRuzza2021} that correlators (cf.~\eqref{eq:correlators} below) of a random Hermitian matrix distributed according to the Jacobi unitary ensemble are generating functions for a type of Hurwitz numbers ({\it triple monotone} Hurwitz numbers); this result extends the combinatorial interpretation of correlators for the Gaussian~\cite{IZ} and Laguerre~\cite{CundenDahlqvistOConnell,GisonniGravaRuzza2020,GLM,HanlonSS} unitary ensembles.

Recently, a deformation of Hurwitz numbers has been constructed in~\cite{ChapuyDolega} (see also~\cite{GouldenJackson}), termed $b$-Hurwitz number, which count (non-orientable) generalized branched coverings of the sphere.
Generating functions of (non-deformed, $b=0$) Hurwitz numbers~\cite{HarnadPaquet,HarnadOrlov,Okounkov} admit explicit expansion in terms of Schur functions (from which it can be shown that they are tau functions of integrable systems), whereas (roughly speaking) for $b$-Hurwitz numbers one replaces Schur functions with Jack symmetric functions.

In a certain sense, this deformation mimics the deformation of unitary-invariant ensembles of random matrices to $\beta$ ensembles, where we always consider the following relation of parameters:
$b = \frac 2 \beta -1$.
For instance, in~\cite{BonzomChapuyDolega} it is shown that an orthogonal (i.e., $\beta=1$) version of the Br\'ezin--Gross--Witten matrix integral is a generating function for monotone $b$-Hurwitz numbers with $b=1$.
The $\beta=2$ version of this fact, relating the Br\'ezin--Gross--Witten integral proper (which is an integral over unitary matrices) to monotone ($b=0$) Hurwitz numbers, is due to Novak~\cite{Novak} (see also~\cite[Section~6.1]{BonzomChapuyDolega}).

Our first aim is to prove another result in this direction: correlators of the Jacobi $\beta$ ensemble are generating functions for (a type) of $b$-Hurwitz numbers.
This recovers the aforementioned result of~\cite{GisonniGravaRuzza2021} when $\beta=2$.
Moreover, it implies analogous results for the Laguerre $\beta$ ensemble, by taking a suitable limit.

Our second aim is to investigate the combinatorial interpretation of a general class of $b$-Hurwitz numbers, namely, the class of multiparametric $b$-Hurwitz numbers with rational weight generating function.
This class covers all the relevant cases for the aforementioned $\beta$~ensembles.
In particular, we will show that all these $b$-Hurwitz numbers count {\it $b$-monotone Hurwitz maps} (introduced in~\cite{BonzomChapuyDolega}) equipped with a special {\it coloring}.
We emphasize the parallel with ordinary (i.e.,~$b=0$) multiparametric Hurwitz numbers with rational weight generating function, which count monotone factorizations into transpositions in the symmetric group equipped with a similar coloring (cf.~Section~\ref{sec:orientable}).

We now proceed to a detailed formulation of the results.

\subsection[Jacobi beta ensemble]{Jacobi $\boldsymbol{\beta}$ ensemble}

The {\it Jacobi $\beta$ ensemble} (of size $n$) is the random point process on the unit interval $(0,1)$ with (almost surely) $n$ particles, the location of which is governed by the joint probability distribution $w^{\sf J}_\beta(\underline x;c,d)\, \mathrm{d} x_1\cdots\mathrm{d} x_n$ given by
\begin{gather*}
w^{\sf J}_\beta(\underline x;c,d) = \frac 1{\mathcal Z_\beta^{\sf J}}\prod_{1\leq i\leq n} \bigl(x_i^{\frac\beta 2 c-1}(1-x_i)^{\frac\beta 2 d-1}\bigr)\!\prod_{1\leq i<j\leq n}\!|x_i-x_j|^\beta ,
\end{gather*}
where $\underline x=(x_1,\dots,x_n)\in (0,1)^n$, $\beta,c,d>0$ and the normalization is explicitly given as
\begin{equation}
\label{eq:Z}
\mathcal Z^{\sf J}_\beta= n!\prod_{1\leq i<j\leq n}\frac{\Gamma\bigl(\tfrac\beta 2(j-i+1)\bigr)}{\Gamma\bigl(\tfrac\beta 2(j-i)\bigr)}\prod_{1\leq i\leq n}\frac{\Gamma\bigl(\tfrac\beta 2(c+n-i)\bigr)\Gamma\bigl(\tfrac\beta 2(d+n-i)\bigr)}{\Gamma\bigl(\tfrac\beta 2(c+d+2n-i-1)\bigr)} ,
\end{equation}
cf.\ Theorem~\ref{thm:ASintegral} below.
It arises as a natural deformation of the case $\beta=2$, the latter being particularly relevant as it describes the eigenvalues of an $n\times n$~random Hermitian matrix~$M$, positive-definite and bounded above by the identity, distributed according to the following unitary-invariant probability measure with Jacobi weight~\cite{Deift,Forrester}:
\begin{equation*}
\frac 1{\mathcal Z'} \det (M)^{c-1} \det(1-M)^{d-1} \, \mathrm{d} M.
\end{equation*}
Here, $\mathrm{d} M$ is the Lebesgue measure on the space of $n\times n$~Hermitian matrices, namely
\begin{equation*}
\mathrm{d} M = \prod_{1\leq i < j\leq n} \mathrm{d} X_{ij} \mathrm{d} Y_{ij} \prod_{1\leq i\leq n}\mathrm{d} X_{ii} ,
\qquad M=X+\mathrm{i} Y,
\end{equation*}
and the normalization is
\begin{equation*}
\mathcal Z' = \frac{\pi^{n(n-1)/2}}{1! 2! \cdots n!} \mathcal Z_{\beta=2}^{\sf J} .
\end{equation*}
This ensemble of random matrices is known as the Jacobi unitary ensemble; analogous models are well-known for $\beta=1$ and $4$ as well, namely the Jacobi orthogonal and symplectic ensembles, respectively~\cite{Forrester}; for general $\beta$, a model of tridiagonal random matrices is given in~\cite{EdelmanSutton} (building on earlier ideas in~\cite{KillipNenciu}).

We shall be in particular interested in the {\it correlators}, defined as the following expectation values, for integers $k_1,\dots,k_\ell$:
\begin{equation}
\label{eq:correlators}
\mathcal C_{k_1,\dots ,k_\ell}^{\sf J}(n,\beta,c,d) = \int_{(0,1)^n} \biggl(\prod_{1\leq i\leq \ell}\big(x_1^{k_i}+\dots+x_n^{k_i}\big)\biggr) w^{\sf J}_\beta(\underline x;c,d)\, \mathrm{d} x_1\cdots\mathrm{d} x_n .
\end{equation}
(In terms of matrix models, this is the expectation of a product of traces of integer powers of the random matrix.)
We will only consider the case where $k_1,\dots,k_\ell$ are all positive or all negative.

\begin{Remark}
It follows from Theorem~\ref{thm:ASintegral} and Corollary~\ref{cor:inverseSelberg} below that $\mathcal C_{k_1,\dots,k_\ell}^{\sf J}(n,\beta,c,d)$ are rational functions of $n$, $\beta$, $c$, $d$.
\end{Remark}

\begin{Remark}
The correlators~\eqref{eq:correlators} appear naturally when expanding the expectations
\begin{equation*}
\int_{(0,1)^n}\prod_{1\leq i\leq\ell}\bigg(\sum_{1\leq j\leq n}\frac{1}{\zeta_i-x_j}\bigg)w^{\sf J}_\beta(\underline x;c,d)\, \mathrm{d} x_1\cdots\mathrm{d} x_n ,\qquad \ell\geq 1 ,\quad \zeta_1,\dots\zeta_\ell\in\mathbb{C}\setminus[0,1] ,
\end{equation*}
as $\zeta_i\to 0,\infty$.
Such expectations play an important role in the study of large-$n$ asymptotics for various statistics of the Jacobi~$\beta$ ensemble via the general theory of {\it loop equations}~\cite{AmbjornMakeenko,BorotGuionnet,ChekhovEynard}, cf.~\cite{FRW} for a recent study in this direction.
It would be interesting to compare the combinatorial results of this paper with large-$n$ limit theorems for the Jacobi~$\beta$ ensemble.
\end{Remark}

\subsection[b-Hurwitz numbers]{$\boldsymbol{b}$-Hurwitz numbers}

Let $\mathcal P$ be the set of all {\it partitions}, i.e., $\lambda\in\mathcal P $ is a weakly decreasing sequence $\lambda=(\lambda_1,\lambda_2,\dots)$ of nonnegative integers which stabilizes to~$0$.
The nonzero $\lambda_i$ are called {\it parts} of~$\lambda$.
We recall the following notations, for a given $\lambda\in\mathcal P $:
\begin{alignat*}{3}
&|\lambda| = \sum_{i\geq 1}\lambda_i ,\qquad && \ell(\lambda) = |\lbrace i\geq 1\colon \lambda_{i}\not=0\rbrace| ,&
\\
&\mathsf{m}_c(\lambda) = |\lbrace i\geq 1\colon \lambda_i=c\rbrace| ,\qquad &&
\mathsf{z}_\lambda = \prod_{c\geq 1}\mathsf{m}_c(\lambda)! c^{\mathsf{m}_c(\lambda)} . &
\end{alignat*}
The {\it diagram} of $\lambda\in\mathcal P$ is the set
\begin{equation*}
\mathsf{D}(\lambda) = \big\lbrace(i,j)\in \mathbb{Z}^2\colon 1\leq i\leq\ell(\lambda),\, 1\leq j\leq \lambda_i\big\rbrace .
\end{equation*}
Elements of $\mathsf{D}(\lambda)$ are customarily denoted $\square$.
Under the involution $(i,j)\mapsto (j,i)$ of $\mathbb{Z}^2$, $\mathsf{D}(\lambda)$ is mapped into the diagram of another partition $\lambda'$, the {\it conjugate partition}.
For any ${\square=(i,j)\in\mathsf{D}(\lambda)}$, we set
\begin{equation*}
\mathsf{arm}_\lambda(\square)=\lambda_i-j ,\qquad
\mathsf{leg}_\lambda(\square)=\lambda_j'-i .
\end{equation*}
The {\it dominance relation} $\preceq_{\sf d}$ is the partial order relation on $\mathcal P $ defined by declaring $\mu\preceq_{\sf d}\lambda$ if and only if $\sum_{i=1}^r\mu_i\leq\sum_{i=1}^r\lambda_i$ for all $r\geq 1$.

Partitions provide a convenient set of labels for various bases of the ring~$\Lambda$ of symmetric functions~\cite[Chapter~I]{Macdonald}.
Concretely, we can think of $\Lambda$ as the ring $\mathbb{C}[\mathbf p]$ of polynomials in infinitely many indeterminates $\mathbf p=(p_1,p_2,p_3,\dots)$, graded by $\deg p_k=k$.
A basis of $\Lambda$ is given by $p_\lambda=p_{\lambda_1}\cdots p_{\lambda_{\ell(\lambda)}}$, for $\lambda\in\mathcal P $.

More precisely, as explained in~\cite[Chapter I]{Macdonald}, $\Lambda$ is the inverse limit as $n\to\infty$ of the inverse system formed by the rings $\Lambda_n=\mathbb{C}_n[x_1,\dots,x_n]^{\mathfrak S_n}$ of symmetric polynomials in $x_1,\dots,x_n$ and by the maps $\Lambda_m\to \Lambda_n$ for $m\geq n$ which send $x_i\mapsto 0$ for $n<i\leq m$.
Then, the variables $p_k$ are the elements of $\Lambda$ that project to the power sum symmetric polynomials $x_1^k+\dots+x_n^k$.

We also need the elements $m_\lambda\in\Lambda$ (for any $\lambda\in\mathcal P $) which project to the monomial symmetric polynomials
\begin{equation*}
\frac 1{\prod_{c\geq 1}\mathsf{m}_c(\lambda)!}\sum_{\sigma\in\mathfrak S_n}x_{\sigma(1)}^{\lambda_1}\cdots x_{\sigma(n)}^{\lambda_n} .
\end{equation*}
The $m_\lambda$ also form a basis of $\Lambda$, called the monomial basis.

Another basis of $\Lambda$ is given by the Jack functions $\mathrm{P}^{(\alpha)}_{\lambda}(\mathbf p)$ (see~\cite[Chapter~VI]{Macdonald} and~\cite{Stanley}), which also depend (rationally) on a parameter~$\alpha>0$.
They reduce to Schur functions when $\alpha=1$ and to Zonal functions when $\alpha=2$.
In general, they are uniquely determined by the following properties.
\begin{itemize}\itemsep=0pt
\item An {\it orthogonality} property:
\begin{equation*}
\bigl\langle \mathrm{P}_\lambda^{(\alpha)},\mathrm{P}_\mu^{(\alpha)}\bigr\rangle_\alpha = 0 ,\qquad \lambda,\mu\in\mathcal P ,\quad \lambda\not=\mu ,
\end{equation*}
where $\langle \,,\,\rangle_\alpha$ is the deformed Hall scalar product on~$\Lambda$, defined by
\begin{equation*}
\bigl\langle p_\lambda, p_\mu\bigr\rangle_\alpha = \delta_{\lambda\mu} \mathsf{z}_\lambda \alpha^{\ell(\lambda)} .
\end{equation*}
\item A {\it triangularity} condition with respect to the monomial basis: if we write
\begin{equation*}
\mathrm{P}_\lambda^{(\alpha)} = \sum_{\mu\in\mathcal P ,\ |\mu|=|\lambda|}v_{\lambda\mu}^{(\alpha)} m_\mu ,
\end{equation*}
we have $v_{\lambda\mu}^{(\alpha)}\not=0$ only when $\mu\preceq_{\sf d}\lambda$.
\item A {\it normalization} condition: $v_{\lambda\lambda}^{(\alpha)}=1\,$.
\end{itemize}

\begin{Remark}
We use the P-normalization, rather then the J-normalization adopted in~\cite{BonzomChapuyDolega,ChapuyDolega}.
To compare the two, one has $\mathrm J_\lambda^{(\alpha)}=\mathsf{h}_\alpha(\lambda)\mathrm{P}^{(\alpha)}_\lambda$, with the notation~\eqref{eq:hook}.
\end{Remark}

Further, we have
\begin{equation*}
\bigl\langle \mathrm{P}_\lambda^{(\alpha)},\mathrm{P}_\lambda^{(\alpha)}\bigr\rangle_\alpha = \frac {\mathsf{h}_\alpha'(\lambda)}{\mathsf{h}_\alpha(\lambda)} ,
\end{equation*}
where
\begin{gather}
\mathsf{h}_\alpha (\lambda)=\prod_{\square\in\mathsf{D}(\lambda)}(\alpha\,\mathsf{arm}_\lambda(\square)+\mathsf{leg}_\lambda(\square)+1),\nonumber\\
\mathsf{h}_\alpha'(\lambda)=\prod_{\square\in\mathsf{D}(\lambda)}(\alpha\,\mathsf{arm}_\lambda(\square)+\mathsf{leg}_\lambda(\square)+\alpha) .\label{eq:hook}
\end{gather}

\medskip

Following~\cite[Section~6]{ChapuyDolega} and specializing to the case of present interest, the multiparametric $b$-Hurwitz numbers associated with a formal power series $G(z)=1+\sum_{i\geq 1}g_iz^i\in\mathbb{C}[[z]]$ are defined through the following formal generating series
\begin{equation}
\label{eq:Jackexp}
\tau_G^b(\epsilon;\mathbf p) = \sum_{\lambda\in\mathcal P }\frac{1}{\mathsf{h}_{b+1}'(\lambda)}\mathrm{P}_\lambda^{(b+1)}(\mathbf p)\prod_{\square\in\mathsf{D}(\lambda)}G(\epsilon\,\mathsf{c}_{b+1}(\square)),
\end{equation}
where, for $\square=(i,j)\in\mathsf{D}(\lambda)$,
\begin{equation*}
\mathsf{c}_{\alpha}(\square) = \alpha(j-1)-(i-1) .
\end{equation*}

\begin{Definition}
\label{def:HN}
The \textit{$b$-Hurwitz number} $H_G^b(\lambda;r)$ is the coefficient of $\epsilon^rp_\lambda$ in $\tau_G^b(\epsilon;\mathbf p)$, namely
\begin{equation}
\label{eq:tau}
\tau_G^b(\epsilon;\mathbf p) = \sum_{\lambda\in\mathcal P }\sum_{r\geq 0}H_G^b(\lambda;r)\epsilon^rp_\lambda .
\end{equation}
\end{Definition}

The geometric meaning of $b$-Hurwitz numbers has been unveiled in~\cite{ChapuyDolega}.
Below we report a~different combinatorial interpretation which closely follows~\cite{BonzomChapuyDolega} instead, see Theorem~\ref{thm:geo} and Section~\ref{sec:geometry}.

\begin{Remark}
The definition in~\cite{ChapuyDolega} is more general.
An additional set of times $\mathbf q=(q_1,q_2,\dots)$ is considered and the Jack expansion~\eqref{eq:Jackexp} is replaced by
\begin{equation}
\label{eq:moregeneral}
\widetilde \tau^b_G(\epsilon;\mathbf p,\mathbf q) = \sum_{\lambda\in\mathcal P }\frac{\mathsf{h}_{b+1}(\lambda)}{\mathsf{h}'_{b+1}(\lambda)} \mathrm{P}_\lambda^{(b+1)}(\mathbf{p}) \mathrm{P}_\lambda^{(b+1)}(\mathbf q) \prod_{\square\in\mathsf{D}(\lambda)}G(\epsilon\,\mathsf{c}_{b+1}(\square))
\end{equation}
which then allows one to define more general $b$-Hurwitz numbers depending on two partitions~$\lambda$ and~$\mu$, by extracting the coefficient in front of $p_\lambda q_\mu$.
The reduction to~\eqref{eq:Jackexp} is performed by setting $q_1=1$ and $q_i=0$ for all~$i\geq 2$, and by using the special value $\mathrm{P}_\lambda^{(\alpha)}(1,0,0,\dots)=1/\mathsf{h}_\alpha(\lambda)$, cf.\
\cite[Chapter~VI, equation~(10.29)]{Macdonald}.
\end{Remark}

\begin{Remark}[connections to integrable systems]
The case $\beta=2$, i.e., $b=0$, falls back to the theory of multiparametric weighted Hurwitz numbers of Guay-Paquet, Harnad, Orlov~\cite{HarnadPaquet,HarnadOrlov}.
It is known in this case that the generating function $\tau^{b=0}_G$ satisfies the KP hierarchy in the times~$\mathbf p$ (in the same case $b=0$, the more general version in~\eqref{eq:moregeneral} satisfies the 2D~Toda hierarchy in the times $\mathbf p$, $\mathbf q$), a far-reaching generalization of Okounkov's seminal result~\cite{Okounkov}.
When $\beta=1$, i.e., $b=1$, a relation to the BKP hierarchy has been established by Bonzom, Chapuy, and Do\l\c{e}ga~\cite{BonzomChapuyDolega}.
\end{Remark}

\subsection[Jacobi beta ensemble and b-Hurwitz numbers]{Jacobi $\boldsymbol{\beta}$ ensemble and $\boldsymbol{b}$-Hurwitz numbers}

For $\lambda\in\mathcal P$, let us introduce the following notation for the correlators in~\eqref{eq:correlators}:
\begin{equation}
\label{eq:correlatorsnew}
\mathcal C_{\lambda}^{\sf J} (n,\beta,c,d ) =
\mathcal C_{\lambda_1,\dots ,\lambda_{\ell(\lambda)}}^{\sf J} (n,\beta,c,d ) ,\qquad
\mathcal C_{-\lambda}^{\sf J} (n,\beta,c,d ) =
\mathcal C_{-\lambda_1,\dots ,-\lambda_{\ell(\lambda)}}^{\sf J} (n,\beta,c,d ) .
\end{equation}

\begin{Theorem}
\label{thm}
Let $\lambda\in\mathcal P $. We have the following Laurent expansions as $n\to+\infty$ of the correlators defined in~\eqref{eq:correlatorsnew}:
\begin{gather}
\label{eq:thmJ1}
\frac 1{\mathsf{z}_\lambda}
\biggl(\frac\beta 2\biggr)^{\ell(\lambda)}\!\!
\biggl(\frac{\gamma+\delta}{\gamma n}\biggr)^{|\lambda|}
\mathcal C_{\lambda}^{\sf J} (n,\beta,c=n(\gamma-1)+1,d=n(\delta-1)+1 ) \!=\! \sum_{r\geq 0}\frac 1{n^r}H^b_{G_+^{\sf J}}(\lambda;r),\!
\\
\label{eq:thmJ2}
\frac 1{\mathsf{z}_\lambda}
\biggl(\frac\beta 2\biggr)^{\ell(\lambda)}\!\!
\biggl(\frac{\gamma}{(\gamma+\delta) n}\biggr)^{|\lambda|}
\mathcal C_{-\lambda}^{\sf J} (n,\beta,c=n\gamma+\tfrac 2\beta,d=n(\delta-1)+1 ) \!=\! \sum_{r\geq 0}\frac 1{n^r}H^b_{G_-^{\sf J}}(\lambda;r),\!
\end{gather}
where $b=\tfrac 2\beta-1$, $\gamma$, $\delta$ are arbitrary complex variables, and
\begin{equation}
\label{eq:G}
G_+^{\sf J}(z) = \frac{(1+z)\big(1+\frac z\gamma\big)}{1+\frac z{\gamma+\delta}} ,\qquad
G_-^{\sf J}(z) = \frac{(1+z)\big(1-\frac{z}{\gamma+\delta}\big)}{1-\frac z\gamma} .
\end{equation}
\end{Theorem}

The proof is contained in Section~\ref{sec:proof}.

\subsection{Laguerre limit}

The {\it Laguerre $\beta$ ensemble} (of size $n$) is the random point process on the positive half-line $(0,+\infty)$ with (almost surely) $n$ particles, the location of which is governed by the joint probability distribution $w^{\sf L}_\beta(\underline y;c)\, \mathrm{d} y_1\cdots\mathrm{d} y_n$ given by
\begin{equation*}
w^{\sf L}_\beta(\underline y;c) = \frac 1{\mathcal Z_\beta^{\sf L}} \prod_{1\leq i\leq n} \bigl(y_i^{\frac\beta 2 c-1} \mathrm{e}^{-y_i}\bigr)\prod_{1\leq i<j\leq n}|y_i-y_j|^\beta ,
\qquad \underline y=(y_1,\dots,y_n)\in (0,+\infty)^n ,
\end{equation*}
where $\beta,c>0$ and the normalization is explicitly given as
\begin{equation}
\label{eq:ZLUE}
\mathcal Z_\beta^{\sf L} = n! \prod_{1\leq i<j\leq n}\frac{\Gamma\bigl(\tfrac\beta 2(j-i+1)\bigr)}{\Gamma\bigl(\tfrac\beta 2(j-i)\bigr)} \prod_{1\leq i\leq n}\Gamma\bigl(\tfrac\beta 2(c+n-i)\bigr) .
\end{equation}
Again, the cases $\beta=1,2,4$ correspond to eigenvalue distributions of well-known random matrix ensembles (respectively, orthogonal, unitary, and symplectic Laguerre ensembles, see, e.g.,~\cite{Forrester}) while for general $\beta$ a model of tridiagonal random matrices has been given in~\cite{DumitriuEdelman}.

We analogously define correlators as the following expectation values, for integers $k_1,\dots,k_\ell$:%
\begin{equation}
\label{eq:correlatorsLUE}
\mathcal C_{k_1,\dots ,k_\ell}^{\sf L}(n,\beta,c) = \int_{(0,+\infty)^n} \biggl(\prod_{1\leq i\leq \ell}\big(y_1^{k_i}+\dots+x_n^{y_i}\big)\biggr) w^{\sf L}_\beta(\underline y;c)\, \mathrm{d} y_1\cdots\mathrm{d} y_n .
\end{equation}

The Laguerre $\beta$~ensemble is a limit of the Jacobi $\beta$~ensemble.
For instance,~\eqref{eq:ZLUE} can be deduced from~\eqref{eq:Z} by a change of integration variables and the limit $d\to+\infty$.
Moreover, we have
\begin{equation*}
\lim_{d\to+\infty}\biggl(\frac\beta 2d\biggr)^{k_1+\dots+k_\ell} \mathcal C_{k_1,\dots ,k_\ell}^{\sf J}(n,\beta,c,d) = \mathcal C_{k_1,\dots ,k_\ell}^{\sf L}(n,\beta,c)
\end{equation*}
for all integers $k_1,\dots, k_\ell$.
Hence, Theorem~\ref{thm} implies the following expansions of Laguerre correlators.

\begin{Theorem}
\label{thm:Laguerre}
Let $\lambda\in\mathcal P $. We have the following Laurent expansions as $n\to+\infty$ of the correlators defined in~\eqref{eq:correlatorsLUE}:
\begin{align}
\label{eq:thmL1}
&\frac 1{\mathsf{z}_\lambda}
\biggl(\frac\beta2\biggr)^{\ell(\lambda)-|\lambda|}
\bigl(\gamma n^2\bigr)^{-|\lambda|}
 \mathcal C_{\lambda_1,\dots ,\lambda_{\ell(\lambda)}}^{\sf L}(n,\beta,c=n(\gamma-1)+1) = \sum_{r\geq 0}\frac 1{n^r} H^b_{G_+^{\sf L}}(\lambda;r) ,
\\
\label{eq:thmL2}
&\frac 1{\mathsf{z}_\lambda}
\biggl(\frac\beta2\biggr)^{\ell(\lambda)+|\lambda|}
\gamma^{|\lambda|}
 \mathcal C_{-\lambda_1,\dots ,-\lambda_{\ell(\lambda)}}^{\sf L}\bigl(n,\beta,c=n\gamma+\tfrac 2\beta\bigr)
 = \sum_{r\geq 0}\frac 1{n^r} H^b_{G_-^{\sf L}}(\lambda;r) ,
\end{align}
where $b=\tfrac 2\beta-1$, $\gamma$ is an arbitrary complex variable, and
\begin{equation*}
G_+^{\sf L}(z) = (1+z)\big(1+\tfrac z\gamma\big) ,\qquad
G_-^{\sf L}(z) = \frac{1+z}{1-\frac z\gamma} .
\end{equation*}
\end{Theorem}

\begin{Remark}
The expansions of Theorem~\ref{thm:Laguerre} are known for $\beta=2$ in terms of ($b=0$)~Hurwitz numbers~\cite{CundenDahlqvistOConnell,GLM,HanlonSS}.
The result for the positive Laguerre correlators is also mentioned in~\cite[Appendix~A]{BonzomChapuyDolega}, in which case, the relevant $b$-Hurwitz numbers are given an equivalent combinatorial meaning in terms of bipartite (possibly non-orientable) maps.
It would be interesting to compare with the results of~\cite{GLMhyper} for the Laguerre orthogonal ensemble (i.e.,~$\beta=1$) and the hyperoctahedral group.
\end{Remark}

\subsection{Colored monotone Hurwitz maps}\label{sec:introgeo}

We now give a combinatorial model for the above expansions, inspired by~\cite{BonzomChapuyDolega}.
We will consider arbitrary rational weight generating functions $G$, covering all cases considered above.

\begin{Remark}
In principle, a combinatorial interpretation for $b$-Hurwitz numbers for arbitrary~$G$ is implicit from~\cite[Theorem~6.2]{ChapuyDolega} in terms of a weighted count of (non-orientable) generalized branched coverings of the sphere, but we prefer to give a (perhaps) more explicit description when $G$ is rational following~\cite{BonzomChapuyDolega}.
\end{Remark}

The notion of monotone non-orientable Hurwitz map has been proposed in~\cite[Section~3]{BonzomChapuyDolega}.
We recall it here for the reader's convenience and refer to loc.\ cit.\ for more details.

\begin{Definition}[monotone Hurwitz maps~\cite{BonzomChapuyDolega}]
A \textit{monotone Hurwitz map} $\Gamma$ is an embedding of a loopless multigraph in a compact (possibly non-orientable) surface with the following properties.
\begin{enumerate}\itemsep=0pt
\item The complement of the multigraph in the surface is homeomorphic to a disjoint union of disks, called \textit{faces}.
\item The $n$ vertices of $\Gamma$ are labeled with the numbers in $\lbrace 1,\dots,n\rbrace$ and the $r$ edges of $\Gamma$ are labeled $e_1,\dots, e_r$.
For each $e_i$ we denote $a_i$, $b_i$ the vertices connected by $e_i$ such that~${a_i<b_i}$.
\item We have $b_1\leq\dots\leq b_r$.
\item A neighborhood of each vertex of $\Gamma$ is equipped with an orientation.
\item Each vertex of $\Gamma$ is equipped with a distinguished sector between consecutive half-edges, which is termed \textit{active corner}.
\item For each $i\in\lbrace 1,\dots,r\rbrace$, let $\Gamma_i$ be obtained from $\Gamma$ by removing the edges $e_{i+1},\dots,e_r$.
Note that $\Gamma_i$ might be embedded in a different compact surface.
In the map $\Gamma_i$ the following conditions must be met:
\begin{itemize}\itemsep=0pt
\item[--] the active corner at $b_i$ immediately follows $e_i$;
\item[--] the active corner at $a_i$ is opposite (with respect to $e_i$) to the active corner at $b_i$;
\item[--] if the edge $e_i$ is disconnecting in $\Gamma_i$, the local orientations at $a_i$, $b_i$ are compatible (i.e., they extend to an orientation of a neighborhood of $e_i$).
\end{itemize}
\end{enumerate}
The \textit{degree} of a face is the number of active corners in that face.
The \textit{profile} of a monotone Hurwitz map is the partition whose parts are the degrees of its faces.
\end{Definition}

See Figure~\ref{fig} for an example.
Note that the definition allows maps containing connected components consisting of a single point embedded in the sphere.

\begin{Remark}
This definition allows one to iteratively construct monotone Hurwitz maps by adding an extra vertex of maximum label and edges incident to it with increasing labels (the possible ways of attaching these new edges are specified by property 6 above).
This key fact is used in the proof of~\cite[Proposition~3.2]{BonzomChapuyDolega}, and we will similarly employ it in the proof of Theorem~\ref{thm:geo} below.
Moreover, it is clear by this inductive construction that any face has at least one active corner.
Therefore, a monotone Hurwitz map with $r$ edges and profile $\lambda$ has Euler characteristic
\begin{equation}
\label{eq:char}
\chi = |\lambda|-r+\ell(\lambda) .
\end{equation}
\end{Remark}

One of the main constructions in~\cite{ChapuyDolega} is the definition of weight of monotone Hurwitz maps\footnote{Actually, in loc.\ cit. the $b$-weight is defined for more general combinatorial objects, namely non-orientable constellations.
The restriction to this case is done in~\cite{BonzomChapuyDolega}.} which is monomial in a variable $b$, termed \textit{b-weight}.

\begin{Definition}[$b$-weight of monotone Hurwitz maps~\cite{BonzomChapuyDolega}]
Let a monotone non-orientable Hurwitz map~$\Gamma$ be given.
We iteratively remove either the vertex of maximum label (if it is isolated) or the edge of maximum label (which is necessarily incident to the vertex of maximum label). Doing so we also record a weight $1$ or $ b$ each time we delete an edge, chosen as follows.
Let $e$ the edge being deleted and $\Gamma'$ the graph after deletion of $e$. The weight is decided according to the following rules:
\begin{itemize}\itemsep=0pt
\item if $e$ joins two active corners in the same face of $\Gamma'$, the weight is $1$ if $e$ splits the face into two, and $b$ otherwise;
\item if $e$ joins two active corners in different faces of $\Gamma'$, the weight is $1$ if the local orientations at the joined vertices are compatible (i.e., they extend to an orientation of a neighborhood of $e$), and $b$ otherwise.
\end{itemize}
The product of all weights recorded during this procedure (performed until the graph is empty) is the \textit{$b$-weight} of $\Gamma$, which is a power of $b$ denoted $b^{\nu(\Gamma)}$.
\end{Definition}

Note that $\nu(\Gamma)=0$ if and only if $\Gamma$ is orientable.

\begin{Remark}
The $b$-weight is a necessary ingredient to give an enumerative interpretation to~\eqref{eq:Jackexp}, cf.\ Theorem~\ref{thm:geo} below.
Remarkably, the $b$-weight is \textit{neither unique nor canonical} (even if the generating function is).
More precisely, such a $b$-weight is only subject to the constraints of a \textit{measure of non-orientability} in the sense of~\cite[Section~3]{ChapuyDolega} which however do not fix it uniquely.

In the definition above, we have exploited the fixed local orientations at the vertices to make a choice in order to uniquely define \textit{a} $b$-weight.
More generally, when the edge $e$ joins two active corners in different faces, the $b$-weight could be arbitrarily chosen to be $1$ or $b$, with the caveat to choose different $b$-weights for different \textit{twists} of the edge $e$ and to choose $1$ in the orientable case.
We content ourself with the practical definition above and refer to~\cite{BonzomChapuyDolega,ChapuyDolega} for more details.
\end{Remark}

We need one last definition.

\begin{Definition}[colored monotone Hurwitz maps]
Let $L,M$ be nonnegative integers.
An \textit{$(L|M)$-coloring} of a monotone Hurwitz map $\Gamma$ is a mapping $c\colon\lbrace 1,\dots,r\rbrace\to\lbrace 1,\dots,L+M\rbrace$ such that:
\begin{itemize}\itemsep=0pt
\item if $1\leq i<j\leq r$ and $1\leq c(i)=c(j)\leq L$, then $b_i<b_j$;
\item if $1\leq i<j\leq r$ and $b_i=b_j$, then $c(i)\leq c(j)$.
\end{itemize}
An \textit{$(L|M)$-colored monotone Hurwitz map} is a monotone Hurwitz map with an $(L|M)$-coloring.
\end{Definition}

For example, a $(0|1)$-colored monotone Hurwitz map is just a monotone Hurwitz map.
See Figure~\ref{fig} for another example.

\begin{Theorem}
\label{thm:geo}
For a set of parameters $u_1,\dots,u_{L+M}$, let
\begin{equation*}
G(z)\,=\,\frac{\prod_{i=1}^L(1+u_iz)}{\prod_{i=1}^M(1-u_{L+i}z)} .
\end{equation*}
Then, for all $\lambda\in\mathcal P $ and $r\geq 0$,
\begin{equation*}
H_G^b(\lambda;r) = \frac 1{|\lambda|!}
\sum_{(\Gamma,c)}\frac{b^{\nu(\Gamma)}}{(1+b)^{|\pi_0(\Gamma)|}} u_{c(1)}\cdots u_{c(r)},
\end{equation*}
where $(\Gamma,c)$ runs in the set of $(L|M)$-colored monotone Hurwitz maps with $r$ edges and profile~$\lambda$ (hence, $|\lambda|$ vertices).
Moreover, $b^{\nu(\Gamma)}$ is the $b$-weight of~$\Gamma$ and $\pi_0(\Gamma)$ the set of connected components of~$\Gamma$.
\end{Theorem}

The proof is contained in Section~\ref{sec:geometry} and is a generalization of the argument in~\cite[Proposition~3.2]{BonzomChapuyDolega}.
We consider the reduction to the orientable case ($b=0$) in Section~\ref{sec:orientable}.

\begin{Remark}
We consider the redundant parameter $\epsilon$ in~\eqref{eq:Jackexp} to make clear the large-$n$ expansion of correlators~\eqref{eq:thmJ1}--\eqref{eq:thmJ2} and~\eqref{eq:thmL1}--\eqref{eq:thmL2}.
Because of~\eqref{eq:char}, these large-$n$ expansions can be interpreted as {\it topological expansions} (over not necessarily orientable geometries) after a~global rescaling of correlators by a power of $n$.
\end{Remark}

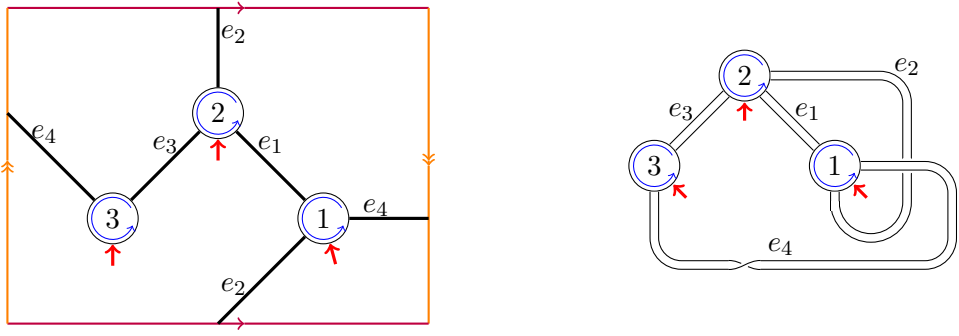
\begin{figure}[t]
\center
\begin{tikzpicture}
\begin{scope}[shift={(-3.5,0)}]
\begin{scope}[scale=.7]

\draw[->,thick,color=purple] (-4,3) -- (.5,3);
\draw[thick,color=purple] (.5,3) -- (4,3);

\draw[->,thick,color=orange] (4,3) -- (4,.1);
\draw[->,thick,color=orange] (4,.1) -- (4,0);
\draw[thick,color=orange] (4,0) -- (4,-3);

\draw[->,thick,color=purple] (-4,-3) -- (.5,-3);
\draw[thick,color=purple] (.5,-3) -- (4,-3);

\draw[->,thick,color=orange] (-4,-3) -- (-4,0);
\draw[->,thick,color=orange] (-4,0) -- (-4,.1);
\draw[thick,color=orange] (-4,.1) -- (-4,3);

\node at (2,-1) [circle,draw] (1) {$1$};
\node at (0,1) [circle,draw] (2) {$2$};
\node at (-2,-1) [circle,draw] (3) {$3$};

\draw[very thick] (3) -- (2) -- (1) -- (4,-1);
\draw[very thick] (-4,1) -- (3);
\draw[very thick] (2) -- (0,3);
\draw[very thick] (1) -- (0,-3);

\node at (1,.4) {$e_1$};
\node at (0.3,2.5) {$e_2$};
\node at (0.3,-2.3) {$e_2$};
\node at (-1,.4) {$e_3$};
\node at (3,-.8) {$e_4$};
\node at (-3.3,.6) {$e_4$};

\begin{scope}[shift={(1)}]
\begin{scope}[rotate=30]
\draw[->,color=blue] (0.4,0) arc (0:310:.4);
\end{scope}
\end{scope}

\begin{scope}[shift={(2)}]
\begin{scope}[rotate=30]
\draw[->,color=blue] (0.4,0) arc (0:310:.4);
\end{scope}
\end{scope}

\begin{scope}[shift={(3)}]
\begin{scope}[rotate=30]
\draw[->,color=blue] (0.4,0) arc (0:310:.4);
\end{scope}
\end{scope}

\begin{scope}[shift={(1)}]
\begin{scope}[rotate=15]
\draw[color=red,->,very thick] (0,-.9) -- (0,-.5);
\end{scope}
\end{scope}

\begin{scope}[shift={(2)}]
\begin{scope}[rotate=0]
\draw[color=red,->,very thick] (0,-.9) -- (0,-.5);
\end{scope}
\end{scope}

\begin{scope}[shift={(3)}]
\draw[color=red,->,very thick] (0,-.9) -- (0,-.5);
\end{scope}

\end{scope}
\end{scope}

\begin{scope}[shift={(3.5,0)}]
\begin{scope}[scale=1.2]

\node at (1,0) [circle,draw] (1) {$1$};
\node at (0,1) [circle,draw] (2) {$2$};
\node at (-1,0) [circle,draw] (3) {$3$};

\begin{scope}[shift={(1)}]
\begin{scope}[rotate=30]
\draw[->,color=blue] (0.23,0) arc (0:310:.23);
\end{scope}
\end{scope}

\begin{scope}[shift={(2)}]
\begin{scope}[rotate=30]
\draw[->,color=blue] (0.23,0) arc (0:310:.23);
\end{scope}
\end{scope}

\begin{scope}[shift={(1)}]
\begin{scope}[rotate=45]
\draw[color=red,->,very thick] (0,-.5) -- (0,-.3);
\end{scope}
\end{scope}

\begin{scope}[shift={(2)}]
\begin{scope}[rotate=0]
\draw[color=red,->,very thick] (0,-.5) -- (0,-.3);
\end{scope}
\end{scope}

\begin{scope}[shift={(3)}]
\begin{scope}[rotate=30]
\draw[->,color=blue] (0.23,0) arc (0:310:.23);
\end{scope}
\end{scope}

\begin{scope}[shift={(3)}]
\begin{scope}[rotate=45]
\draw[color=red,->,very thick] (0,-.5) -- (0,-.3);
\end{scope}
\end{scope}

\draw[double distance=1mm] (3) -- (2) -- (1);

\draw[double distance=1mm, line cap=white] (2) -- (1.5,1) arc (90:0:.3) -- (1.8,-.4) arc (0:-180:.4) -- (1,-.5) -- (1);

\draw[fill=white, color=white] (1.8,0) circle (2.5pt);

\draw[double distance=1mm, line cap=white] (1) -- (2,0) arc (90:0:.3) -- (2.3,-.8) arc (0:-90:.3) -- (-.7,-1.1) arc (-90:-180:.3) -- (3);

\draw[fill=white, color=white] (0,-1.1) circle (5pt);

\node at (.7,.6) {$e_1$};
\node at (1.8,1.1) {$e_2$};
\node at (-.7,.6) {$e_3$};
\node at (.4,-.9) {$e_4$};

\draw (5pt,-1.15) to[out=180, in=0] (-5pt,-1.05);
\draw[fill=white, color=white] (0,-1.1) circle (1pt);
\draw (5pt,-1.05) to[out=180, in=0] (-5pt,-1.15);

\end{scope}
\end{scope}
\end{tikzpicture}
\caption{On the left, a monotone Hurwitz map with $n=3$ vertices and $r=4$ edges embedded in the Klein bottle.
Active corners are indicated by red arrows and local orientation around vertices are indicated by blue arrows.
It has $b$-weight $b$ and profile $(3)$.
On the right, the same non-orientable labeled Hurwitz map depicted as a \textit{ribbon graph}, which is a small open neighborhood of the graph in the surface.
For this example, the mapping $c(1)=1=c(3)$, $c(2)=2=c(4)$ is simultaneously a $(2|0)$-, $(1|1)$-, or $(0|2)$-coloring; the mapping $c(1)=1=c(2)$, $c(3)=2=c(4)$ is a $(0|2)$-coloring.}
\label{fig}
\end{figure}

\section[Proof of Theorem~\ref{thm}]{Proof of Theorem~\ref{thm}}\label{sec:proof}

In this section, we will denote
\begin{equation*}
\widetilde{\mathrm{P}}_\lambda^{(\alpha)}(x_1,\dots,x_n) = \mathrm{P}_\lambda^{(\alpha)}(\mathbf p)|_{p_k=\sum_{i=1}^nx_i^k} .
\end{equation*}

A crucial ingredient in the proof is the following formula for the expectation of Jack polynomials with respect to the Jacobi weight which is due to Kadell~\cite{Kadell} (see also~\cite[Conjecture~C5]{MacdonaldAG} and~\cite[Chapter~VI.10]{Macdonald}) and which generalizes the celebrated Selberg integral~\cite{Selberg}.

\begin{Theorem}[\cite{Kadell}]
\label{thm:ASintegral}
For all $\beta,c,d>0$ and all $\lambda\in\mathcal P $, we have
\begin{align*}
&\frac 1{n!}\int_{(0,1)^n}\widetilde{\mathrm{P}}_\lambda^{(2/\beta)}(\underline x) \prod_{1\leq i\leq n} \bigl(x_i^{\frac\beta 2 c-1} (1-x_i)^{\frac\beta 2 d-1}\bigr) \prod_{1\leq i<j\leq n}|x_i-x_j|^\beta\, \mathrm{d} x_1\cdots\mathrm{d} x_n
\\
&\qquad =
\prod_{1\leq i<j\leq n}\frac{\Gamma\bigl(\lambda_i-\lambda_j+\tfrac\beta 2(j-i+1)\bigr)}{\Gamma\bigl(\lambda_i-\lambda_j+\tfrac\beta 2(j-i)\bigr)}
\prod_{1\leq i\leq n}\frac{\Gamma\bigl(\lambda_i+\tfrac \beta2(c+n-i)\bigr) \Gamma\bigl(\tfrac \beta2(d+n-i)\bigr)}{\Gamma\bigl(\lambda_i+\tfrac \beta2(c+d+2n-i-1)\bigr)} .
\end{align*}
\end{Theorem}

When $\lambda$ is the empty partition, i.e.,~$\lambda=(0,0,\dots)$, we obtain the value for $\mathcal Z_\beta^{\sf J}$ claimed in~\eqref{eq:Z}.

To consider \emph{negative} correlators, a further property is needed. Let us introduce the notation
\begin{equation*}
\underline x^{-1}=\big(x_1^{-1},\dots,x_n^{-1}\big) .
\end{equation*}

\begin{Lemma}
\label{lemma:inverse}
Let $n\geq 0$ be an integer and let $\lambda\in\mathcal P$ be such that $\ell(\lambda)\leq n$.
Define another partition\footnote{Geometrically, the complement of the diagram of $\lambda$ in the rectangle $\lbrace (i,j)\in\mathbb{Z}^2\colon 1\leq i\leq n,\, 1\leq j\leq\lambda_1\rbrace$ is (up to a rotation) the diagram of $\widehat\lambda$.} $\widehat\lambda\in\mathcal P$ by $\widehat\lambda_i=\lambda_1-\lambda_{n-i+1}$.
Then
\begin{equation*}
\widetilde P_\lambda^{(\alpha)}\big(\underline x^{-1}\big) = (x_1\cdots x_n)^{-\lambda_1} \widetilde P_{\widehat\lambda}^{(\alpha)}(\underline x).
\end{equation*}
\end{Lemma}
\begin{proof}
It is known~\cite[Theorem~3.1]{Stanley} that $\widetilde P_\lambda^{(\alpha)}(\underline x)$ are uniquely characterized as the eigenvectors of the {\it Calogero--Sutherland operator}
\begin{equation*}
\mathscr H^{(\alpha)} = \frac \alpha 2 \sum_{i=1}^n D_{x_i}^2 + \frac 12 \sum_{1\leq i<j\leq n}\frac{x_i+x_j}{x_i-x_j} (D_{x_i}-D_{x_j}) ,\qquad D_{z} = z \frac {\partial}{\partial z},
\end{equation*}
also satisfying the triangularity and normalization properties mentioned above.
Namely, they are uniquely characterized by
\begin{equation*}
\mathscr H^{(\alpha)}\ \widetilde{\mathrm{P}}_\lambda^{(\alpha)}(\underline x) = \mathcal E_\lambda^{(\alpha)} \widetilde{\mathrm{P}}_\lambda^{(\alpha)}(\underline x) ,\qquad
\mathcal E_\lambda^{(\alpha)} = \sum_{i=1}^n\biggl(\frac \alpha 2\lambda_i^2+\frac{n+1-2i}{2}\biggr) ,
\end{equation*}
and the fact that $\widetilde{\mathrm{P}}_\lambda^{(\alpha)}(\underline x)$ is a linear combination of monomial symmetric functions of $x_1,\dots, x_n$ associated with partitions smaller than $\lambda$ in the dominance relation $\preceq_{\sf d}$.

It is clear that $\widetilde{\mathrm{P}}_\lambda^{(\alpha)}\big(\underline x^{-1}\big)(x_1\cdots x_n)^{\lambda_1}$ is a symmetric polynomial of $x_1,\dots,x_n$ satisfying the same triangularity and normalization conditions as $\widetilde P_{\widehat \lambda}^{(\alpha)}(\underline x)$.
Moreover, a direct computation shows that
\begin{equation*}
\mathscr H^{(\alpha)}\big(\widetilde{\mathrm{P}}_\lambda^{(\alpha)}\big(\underline x^{-1}\big)(x_1\cdots x_n)^{\lambda_1}\big) = \biggl(\mathcal E_{\lambda}^{(\alpha)}+\frac \alpha2 n\lambda_1^2-\alpha\lambda_1|\lambda|\biggr) \widetilde{\mathrm{P}}_\lambda^{(\alpha)}\big(\underline x^{-1}\big)(x_1\cdots x_n)^{\lambda_1} ,
\end{equation*}
and it is straightforward to check that $\mathcal E_{\lambda}^{(\alpha)}+\frac \alpha2 n\lambda_1^2-\alpha\lambda_1|\lambda|=\mathcal E^{(\alpha)}_{\widehat \lambda}$.
Hence, we conclude that $\widetilde{\mathrm{P}}_\lambda^{(\alpha)}\bigl(\underline x^{-1}\bigr)(x_1\cdots x_n)^{\lambda_1}=\widetilde{\mathrm P}_{\widehat \lambda}^{(\alpha)}(\underline x)$.
\end{proof}

\begin{Corollary}
\label{cor:inverseSelberg}
For all $\beta,d>0$, $\lambda\in\mathcal P $, and $c>\tfrac 2\beta\lambda_1$, we have
\begin{align*}
&\frac 1{n!}\int_{(0,1)^n}\widetilde{\mathrm{P}}_\lambda^{(2/\beta)}\big(\underline x^{-1}\big) \prod_{1\leq i\leq n} \bigl(x_i^{\frac\beta 2 c-1} (1-x_i)^{\frac\beta 2 d-1}\bigr) \prod_{1\leq i<j\leq n}|x_i-x_j|^\beta \, \mathrm{d} x_1\cdots \mathrm{d} x_n
\\
& \qquad =
\prod_{1\leq i<j\leq n}\frac{\Gamma\bigl(\lambda_i-\lambda_j+\tfrac\beta 2(j-i+1)\bigr)}{\Gamma\bigl(\lambda_i-\lambda_j+\tfrac\beta 2(j-i)\bigr)}
\prod_{1\leq i\leq n}\frac{\Gamma\bigl(-\lambda_i+\tfrac \beta2(c+i-1)\bigr) \Gamma\bigl(\tfrac \beta2(d+n-i)\bigr)}{\Gamma\bigl(-\lambda_i+\tfrac \beta2(c+d+n+i-2)\bigr)} .
\end{align*}
\end{Corollary}
\begin{proof}
According to Lemma~\ref{lemma:inverse}, we have
\begin{align*}
&\int_{(0,1)^n}\widetilde{\mathrm{P}}_\lambda^{(2/\beta)}\big(\underline x^{-1}\big) \prod_{1\leq i\leq n} \bigl(x_i^{\frac\beta 2 c-1} (1-x_i)^{\frac\beta 2 d-1}\bigr) \prod_{1\leq i<j\leq n}|x_i-x_j|^\beta \, \mathrm{d} x_1\cdots\mathrm{d} x_n
\\
& \qquad=\int_{(0,1)^n}\widetilde{\mathrm{P}}_{\widehat \lambda}^{(2/\beta)}(\underline x) \prod_{1\leq i\leq n} \bigl(x_i^{\frac\beta 2 (c-\frac 2\beta\lambda_1)-1} (1-x_i)^{\frac\beta 2 d-1}\bigr) \prod_{1\leq i<j\leq n}|x_i-x_j|^\beta \, \mathrm{d} x_1\cdots\mathrm{d} x_n
\end{align*}
such that it suffices to apply Theorem~\ref{thm:ASintegral} after replacing $c$ by $c-\frac 2\beta\lambda_1$ and $\lambda$ by $\widehat\lambda$.
\end{proof}

We consider an infinite sequence of variables $\mathbf t=(t_1,t_2,\dots)$ and denote $t_\lambda=t_{\lambda_1}\cdots t_{\lambda_{\ell(\lambda)}}$ for all $\lambda\in\mathcal P $.
Let us introduce the following {\it formal} generating functions of correlators
\begin{align}
\nonumber
\Upsilon_\beta^\pm(\mathbf t;n)& = \sum_{\lambda\in\mathcal P }\mathcal C^{\sf J}_{\pm\lambda_1,\dots,\pm\lambda_{\ell(\lambda)}}(n,\beta,c,d) \biggl(\frac\beta 2\biggr)^{\ell(\lambda)} \frac{t_\lambda}{\mathsf{z}_\lambda}
\\
\label{eq:genfun}
& = \int_{(0,1)^n}\exp\biggl(\frac\beta 2\sum_{k\geq 1}\frac{t_k}{k}\big(x_1^{\pm k}+\dots+x_n^{\pm k}\big)\biggr) w^{\sf J}_\beta(\underline x;c,d)\,\mathrm{d} x_1\cdots\mathrm{d} x_n ,
\end{align}
the last equality stemming from the well-known identity (cf.~\cite[Chapter~I.2]{Macdonald})
\begin{equation*}
\exp\biggl(\sum_{k\geq 1}\frac{t_k}{k}\big(x_1^k+\dots+x_n^k\big)\biggr) = \sum_{\lambda\in\mathcal P }\frac{t_\lambda}{\mathsf{z}_\lambda}\Biggl(\sum_{i=1}^nx_i^{\lambda_1}\Biggr)\cdots \Biggl(\sum_{i=1}^nx_i^{\lambda_{\ell(\lambda)}}\Biggr) .
\end{equation*}
Let us also recall the Cauchy identity for Jack polynomials~\cite[formula~(6)]{Stanley}
\begin{equation}
\label{eq:Cauchy}
\exp\biggl(\frac 1\alpha \sum_{k\geq 1}\frac {t_k}k \big(x_1^k+\dots+x_n^k\big)\biggr) = \sum_{\lambda\in\mathcal P ,\ \ell(\lambda)\leq n}\frac{\mathsf{h}_\alpha(\lambda)}{\mathsf{h}'_\alpha(\lambda)} \mathrm{P}_\lambda^{(\alpha)}(\mathbf t) \widetilde{\mathrm{P}}_\lambda^{(\alpha)}(x_1,\dots,x_n) .
\end{equation}
\begin{Lemma}
\label{Lemma}
We have
\begin{equation}
\label{eq:holds}
\Upsilon_\beta^\pm(\mathbf t;n) = \sum_{\lambda\in\mathcal P ,\ \ell(\lambda)\leq n}\frac 1{\mathsf{h}_{2/\beta}'(\lambda)} f_\lambda^\pm \mathrm{P}^{(2/\beta)}_\lambda(\mathbf t) ,
\end{equation}
with coefficients
\begin{align*}
&f_\lambda^+ = \prod_{\square\in\mathsf{D}(\lambda)}\frac{ (n+\mathsf{c}_{2/\beta}(\square) ) (c-1+n+\mathsf{c}_{2/\beta}(\square) )}{c+d-2+2n+\mathsf{c}_{2/\beta}(\square)} ,
\\
&f_\lambda^- = \prod_{\square\in\mathsf{D}(\lambda)}\frac{}{}\frac{ (n+\mathsf{c}_{2/\beta}(\square) ) (c+d+n-1-\tfrac 2\beta-\mathsf{c}_{2/\beta}(\square) )}{c-\tfrac 2\beta-\mathsf{c}_{2/\beta}(\square)} .
\end{align*}
\end{Lemma}
\begin{proof}
By Theorem~\ref{thm:ASintegral} and~\eqref{eq:Cauchy}, we infer that~\eqref{eq:holds} holds, in the $+$~case, with
\begin{gather*}
f_\lambda^+ =
\frac{n! \mathsf{h}_{2/\beta}(\lambda)}{\mathcal Z_\beta^{\sf J}}\prod_{1\leq i<j\leq n}\frac{\Gamma\bigl(\lambda_i-\lambda_j+\tfrac\beta 2(j-i+1)\bigr)}{\Gamma\bigl(\lambda_i-\lambda_j+\tfrac\beta 2(j-i)\bigr)} \\
\hphantom{f_\lambda^+ =}{}
\times \prod_{i=1}^n\frac{\Gamma\bigl(\lambda_i+\tfrac \beta2(c+n-i)\bigr) \Gamma\bigl(\tfrac \beta2(d+n-i)\bigr)}{\Gamma\bigl(\lambda_i+\tfrac \beta2(c+d+2n-i-1)\bigr)} .
\end{gather*}
Recall the Pochhammer symbol
\begin{equation*}
(z)_m = \frac{\Gamma(m+z)}{\Gamma(z)} = \prod_{j=1}^m(z+j-1) ,
\end{equation*}
for nonnegative integer $m$.
We use~\eqref{eq:Z} to rewrite
\begin{equation*}
f_\lambda^+ = \mathsf{h}_{2/\beta}(\lambda) \prod_{1\leq i<j\leq n}\frac{\bigl(\tfrac\beta 2(j-i+1)\bigr)_{\lambda_i-\lambda_j}}{\bigl(\tfrac\beta 2(j-i)\bigr)_{\lambda_i-\lambda_j}} \prod_{1\leq i\leq n}\frac{\bigl(\tfrac \beta2(c+n-i)\bigr)_{\lambda_i}}{\bigl(\tfrac \beta2(c+d+2n-i-1)\bigr)_{\lambda_i}} .
\end{equation*}
We reason separately for the two products. The second one is
\begin{align}
\nonumber
\prod_{1\leq i\leq n}\frac{\bigl(\tfrac \beta2(c+n-i)\bigr)_{\lambda_i}}{\bigl(\tfrac \beta2(c+d+2n-i-1)\bigr)_{\lambda_i}}
& = \prod_{1\leq i\leq \ell(\lambda)}\prod_{1\leq j\leq \lambda_i}\frac{\tfrac \beta2(c+n-i)+j-1}{\tfrac \beta2(c+d+2n-i-1)+j-1}
\\
\label{eq:similarly}
& = \prod_{\square\in\mathsf{D}(\lambda)}\frac{c-1+n+\mathsf{c}_{2/\beta}(\square)}{c+d-2+2n+\mathsf{c}_{2/\beta}(\square)} .
\end{align}
About the first one instead, we claim that, due to many cancellations, we have the identity\footnote{Cf.~\cite[Theorem~5.17.1]{Etingof} for the well-known case $\beta=2$, where $\dim\lambda$ is computed in terms of the hook-length formula from the Frobenius determinant formula.}
\begin{equation}
\label{eq:again}
\prod_{1\leq i<j\leq n}\frac{\bigl(\tfrac\beta 2(j-i+1)\bigr)_{\lambda_i-\lambda_j}}{\bigl(\tfrac\beta 2(j-i)\bigr)_{\lambda_i-\lambda_j}}
 = \frac {\prod_{\square\in\mathsf{D}(\lambda)}(n+\mathsf{c}_{2/\beta}(\square))}{\mathsf{h}_{2/\beta}(\lambda)} ,
\end{equation}
which would complete the proof for $f^+_\lambda$.
In the interest of clarity, this claim is proved below, see Appendix~\ref{appendix}.

Similarly, thanks to Corollary~\ref{cor:inverseSelberg} and~\eqref{eq:Cauchy} we infer that~\eqref{eq:holds} holds, in the $-$~case, with $f^-_\lambda$ equal to
\begin{align*}
&\frac{n! \mathsf{h}_{2/\beta}(\lambda)}{\mathcal Z_\beta^{\sf J}}\prod_{1\leq i<j\leq n}\frac{\Gamma\bigl(\lambda_i-\lambda_j+\tfrac\beta 2(j-i+1)\bigr)}{\Gamma\bigl(\lambda_i-\lambda_j+\tfrac\beta 2(j-i)\bigr)}\prod_{i=1}^n\frac{\Gamma\bigl(-\lambda_i+\tfrac \beta2(c+i-1)\bigr) \Gamma\bigl(\tfrac \beta2(d+n-i)\bigr)}{\Gamma\bigl(-\lambda_i+\tfrac \beta2(c+d+n+i-2)\bigr)}
\\
&\qquad = \mathsf{h}_{2/\beta}(\lambda) \prod_{1\leq i<j\leq n}\frac{\bigl(\tfrac\beta 2(j-i+1)\bigr)_{\lambda_i-\lambda_j}}{\bigl(\tfrac\beta 2(j-i)\bigr)_{\lambda_i-\lambda_j}} \prod_{1\leq i\leq n}\frac{\bigl(\tfrac \beta2(c+d+n+i-2)-\lambda_i\bigr)_{\lambda_i}}{\bigl(\tfrac \beta2(c+i-1)-\lambda_i\bigr)_{\lambda_i}} .
\end{align*}
The first part of this expression is again~\eqref{eq:again}, whereas the last part can be rewritten similarly as~\eqref{eq:similarly}:
\begin{align*}
\prod_{1\leq i\leq n}\frac{\bigl(\tfrac \beta2(c+d+n+i-2)-\lambda_i\bigr)_{\lambda_i}}{\bigl(\tfrac \beta2(c+i-1)-\lambda_i\bigr)_{\lambda_i}}
& = \prod_{1\leq i\leq \ell(\lambda)}\prod_{1\leq j\leq \lambda_i}\frac{\tfrac \beta2(c+d+n+i-2)-j}{\tfrac \beta2(c+i-1)-j}
\\
& = \prod_{\square\in\mathsf{D}(\lambda)}\frac{c+d+n-1-\tfrac 2\beta-\mathsf{c}_{2/\beta}(\square)}{c-\tfrac 2\beta-\mathsf{c}_{2/\beta}(\square)},
\end{align*}
and the proof is complete.
\end{proof}

\begin{proof}[Proof of Theorem~\ref{thm}]
By Lemma~\ref{Lemma} and~\eqref{eq:Jackexp}, we obtain the following identities of generating functions:
\begin{align*}
&\Upsilon^+_\beta(\mathbf t;n)\big|_{c=n(\gamma-1)+1,\,d=n(\delta-1)+1} =
\tau_{G_+^{\sf J}}^{\frac 2\beta-1}\bigl(\tfrac 1n;\mathbf p\bigr)\Big|_{p_k=\bigl(\frac{\gamma n}{\gamma+\delta}\bigr)^kt_k} ,\\
&\Upsilon^-_\beta(\mathbf t;n)\big|_{c=n\gamma+\frac 2\beta,\,d=n(\delta-1)+1} =
\tau_{G_-^{\sf J}}^{\frac 2\beta-1}\bigl(\tfrac 1n;\mathbf p\bigr)\Big|_{p_k=\bigl(\frac{(\gamma+\delta) n}{\gamma}\bigr)^kt_k} ,
\end{align*}
where $G_\pm^{\sf J}$ are given in~\eqref{eq:G}.
\big(Note that the restriction $\ell(\lambda)\leq n$ in the sum in~\eqref{eq:holds} can be lifted because $f^\pm_\lambda$ are automatically zero when $\ell(\lambda)>n$.\big)
To complete the proof it suffices to compare the expansions~\eqref{eq:tau} and~\eqref{eq:genfun}.
\end{proof}

\section{Colored monotone Hurwitz maps}\label{sec:geometry}

In this section, we give a combinatorial model for the $b$-Hurwitz numbers of Definition~\ref{def:HN}, when $G$ is (the expansion of) a rational function, in terms of colored monotone Hurwitz maps, introduced in Section~\ref{sec:introgeo}.

\begin{Theorem}
\label{thm:HM}
Let $g_1,\dots,g_{L+M}$ be formal parameters and let
\begin{equation*}
\mathcal T(t,\mathbf p;b) = \sum_{\lambda\in\mathcal P } \frac{t^{|\lambda|}}{\mathsf{h}_{b+1}'(\lambda)} \mathrm{P}_\lambda^{(b+1)}(\mathbf p) \prod_{\square\in\mathsf{D}(\lambda)}G(\mathsf{c}_{b+1}(\square)) ,\qquad
G(z) = \frac{\prod_{i=1}^L(1+g_iz)}{\prod_{i=1}^M(1-g_{L+i}z)} .
\end{equation*}
We have
\begin{equation*}
\mathcal T(t,\mathbf p;b) = \sum_{\lambda\in\mathcal P }\frac{t^{|\lambda|}}{|\lambda|!} p_\lambda
\sum_{(\Gamma,c)}\frac{b^{\nu(\Gamma)}}{(1+b)^{|\pi_0(\Gamma)|}}\prod_{i=1}^{L+M}g_i^{|c^{-1}(i)|},
\end{equation*}
where $(\Gamma,c)$ runs in the set of $(L|M)$-colored monotone Hurwitz maps with profile $\lambda$.
\end{Theorem}

\begin{proof}
The argument is similar to the one in the proof of~\cite[Proposition~3.2]{BonzomChapuyDolega}, to which we refer for more details.
In this proof we write ``CMHM'' in place of ``$(L|M)$-colored monotone Hurwitz map''.

The key property is that a CMHM with $n$ vertices can always be obtained, in a unique way, from a CMHM with $n-1$ vertices by first adding an extra isolated vertex of label $n$, then adding a number of edges incident to $n$ of each color $\lbrace 1,\dots,L+M\rbrace$; the second property in the definition of coloring is ensured by adding edges with increasing label and (weakly) increasing color, and the first property by adding $0$ or $1$ edges (and no more) of color $\leq L$.
We have to track this process at the level of generating functions.

Introduce the generating function, for $n\geq 0$,
\begin{equation*}
\mathcal T^{[n]}(\mathbf p;b) = \sum_{\substack{\lambda\in\mathcal P \\|\lambda|=n}}p_\lambda\sum_{(\Gamma,c)}\frac{b^{\nu(\Gamma)}}{(1+b)^{|\pi_0(\Gamma)|}}\prod_{i=1}^{L+M}g_i^{|c^{-1}(i)|},
\end{equation*}
where $(\Gamma,c)$ runs in the set of CMHM with profile~$\lambda$ (hence, with $n$ vertices).

We have
\begin{equation*}
\frac{z_1}{1+b}\mathcal T^{[n-1]}(\mathbf p;b) = \sum_{\substack{\lambda\in\mathcal P \\|\lambda|=n-1}}z_1 p_\lambda\sum_{(\Gamma,c)}\frac{b^{\nu(\Gamma)}}{(1+b)^{|\pi_0(\Gamma)|}}\prod_{i=1}^{L+M}g_i^{|c^{-1}(i)|},
\end{equation*}
where $(\Gamma,c)$ runs in the set of CMHM with $n$ vertices, where the $n$th vertex is isolated and whose profile is~$\lambda\sqcup \lbrace 1\rbrace$, the extra $1$ being the degree of the face containing the $n$th vertex.

We next attach either $0$ or $1$ edges with color $1$ incident to the $n$th vertex: the argument in~\cite[Proposition~3.2]{BonzomChapuyDolega} shows that\footnote{We have $y_{i}=z_{i+1}$ with respect to the notation in loc.\ cit.}
\begin{equation*}
(1+g_1\Lambda)\frac{z_1}{1+b}\mathcal T^{[n-1]}(\mathbf p;b) = \sum_{j\geq 1}\sum_{\substack{\lambda\in\mathcal P \\|\lambda|=n-j}}z_j p_\lambda\sum_{(\Gamma,c)}\frac{b^{\nu(\Gamma)}}{(1+b)^{|\pi_0(\Gamma)|}}\prod_{i=1}^{L+M}g_i^{|c^{-1}(i)|},
\end{equation*}
where now $(\Gamma,c)$ runs in the set of CMHM with $n$ vertices, where edges incident to the $n$th vertex only have color $1$, and whose profile is~$\lambda\sqcup \lbrace j\rbrace$, $j$ being the degree of the face containing the active corner at $n$.
Indeed, it is proven in loc. cit. that the action of adding an edge incident to the last vertex~$n$ to a Hurwitz map and counting the degree of the face containing the active corner at~$n$ with variables~$z_i$ instead of~$p_i$ corresponds to the action of the operator~$\Lambda$ given by
\begin{equation*}
\Lambda =
(1+b)\sum_{i,j\geq 1}i z_{i+j} \frac{\partial^2}{\partial p_i \partial z_j} +
\sum_{i,j\geq 1}z_i p_j \frac{\partial}{\partial z_{i+j}} +
b\sum_{i\geq 2}(i-1) z_i \frac{\partial}{\partial z_i} .
\end{equation*}

By iterating this argument $L$ times, we obtain
\begin{equation*}
(1+g_L\Lambda)\cdots(1+g_1\Lambda)\frac{z_1}{1+b}\mathcal T^{[n-1]}(\mathbf p;b) = \sum_{j\geq 1}\sum_{\substack{\lambda\in\mathcal P \\|\lambda|=n-j}}z_j p_\lambda\sum_{(\Gamma,c)}\frac{b^{\nu(\Gamma)}}{(1+b)^{|\pi_0(\Gamma)|}}\prod_{i=1}^{L+M}g_i^{|c^{-1}(i)|},
\end{equation*}
where now $(\Gamma,c)$ runs in the set of CMHM with $n$ vertices, where edges incident to the $n$th vertex have colors $\leq L$ only, and whose profile is~$\lambda\sqcup \lbrace j\rbrace$, $j$ being the degree of the face containing the active corner at $n$.

Next, we attach edges incident to the $n$th vertex of color $L+1$; by the same argument as before, taking into account that we can now add as many such edges as we wish, we get
\begin{align*}
&\sum_{p=0}^{+\infty}(g_{L+1}\Lambda)^p (1+g_L\Lambda) \cdots (1+g_1\Lambda) \frac{z_1}{1+b} \mathcal T^{[n-1]}(\mathbf p;b)
\\
& \qquad = \frac 1{1-g_{L+1}\Lambda} (1+g_L\Lambda) \cdots (1+g_1\Lambda) \frac{z_1}{1+b}\mathcal T^{[n-1]}(\mathbf p;b)
\\
&\qquad = \sum_{j\geq 1}\sum_{\substack{\lambda\in\mathcal P \\|\lambda|=n-j}}z_j p_\lambda\sum_{(\Gamma,c)}\frac{b^{\nu(\Gamma)}}{(1+b)^{|\pi_0(\Gamma)|}}\prod_{i=1}^{L+M}g_i^{|c^{-1}(i)|},
\end{align*}
where now $(\Gamma,c)$ runs in the set of CMHM with $n$ vertices, where edges incident to the $n$th vertex only have colors $\leq L+1$, and whose profile is~$\lambda\sqcup \lbrace j\rbrace$, $j$ being the degree of the face containing the active corner at $n$.

Finally, iterating this step $M$-times (and restoring the variables $p_i$ for the degree of the face containing active corner at the $n$th vertex) we get
\begin{equation*}
\mathcal T^{[n]}(\mathbf p;b) = G(\Lambda) \frac{z_1}{1+b} \mathcal T^{[n-1]}(\mathbf p;b) \bigg|_{z_i=p_i} .
\end{equation*}

On the other hand, we know from identity~\cite[equation~(62)]{ChapuyDolega} that
\begin{equation*}
\frac{\partial}{\partial t}\mathcal T(t,\mathbf p;b) = G(\Lambda) \frac{z_1}{1+b} \mathcal T(t,\mathbf p;b) \bigg|_{z_i=p_i} ,
\end{equation*}
and thus we can show (inductively in powers of $t$) that $\mathcal T(t,\mathbf p;b)=\sum_{n\geq 0}\frac {t^n}{n!}\mathcal T^{[n]}(\mathbf p;b)$.
\end{proof}

\begin{proof}[Proof of Theorem~\ref{thm:geo}]
It suffices to note that, by Definition~\ref{def:HN}, when $G(z)\!=\!\frac{\prod_{i=1}^L(1+u_iz)}{\prod_{i=1}^M(1-u_{L+i}z)}$, $H_G^b(\lambda;r)$ is the coefficient of $\epsilon^rp_\lambda$ in $\mathcal T(t,\mathbf p;b)\big|_{t=1,\, g_i=\epsilon u_i}$, where $\mathcal T(t,\mathbf p;b)$ is defined in Theo\-rem~\ref{thm:HM}.\looseness=1
\end{proof}

\subsection{Comparison with the orientable case}\label{sec:orientable}

To illustrate the reduction to the orientable case, we recall the group theoretical interpretation of the Hurwitz numbers $H_G^{b=0}(\lambda;r)$.

Let $\mathfrak S_n$ be the group of permutations of $\lbrace 1,\dots,n\rbrace$ and let, for all $\lambda\in\mathcal P $ with $|\lambda|=n$, $\mathfrak C(\lambda)\subseteq\mathfrak S_n$ be the conjugacy class of permutations whose disjoint cycles have lengths equal to the parts of $\lambda$.
Let $\mathbb{C}[\mathfrak S_n]$ be the group algebra of $\mathfrak S_n$, and $\mathscr J_i\in\mathbb{C}[\mathfrak S_n]$ for $i=1,\dots, n$, be the Young--Jucys--Murphy elements, namely
\begin{equation}
\label{eq:J}
\mathscr J_1=0,\qquad \mathscr J_2=(1,2),\qquad \dots,\qquad \mathscr J_n=(1,n)+(2,n)+\dots+(n-1,n) .
\end{equation}
(We denote $(a,b)\in\mathfrak S_n$ the {\it transposition} that exchanges $a$, $b$.)
Let us recall that although the elements $\mathscr J_i$ are not central, they form a (maximal) commutative subalgebra.
Moreover, any symmetric polynomial in $n$ variables evaluated at $\mathscr J_1,\dots,\mathscr J_n$ is central, hence it can be expanded (uniquely) as a linear combination of{\samepage
\begin{equation*}
\mathscr C_\lambda =\sum_{\sigma\in\mathfrak C(\lambda)}\sigma
\end{equation*}
for $\lambda\in\mathcal P $ with $|\lambda|=n$.}

It is well-known, cf.~\cite{HarnadPaquet} or, for a review,~\cite[Sections~2.1 and 2.2]{GisonniGravaRuzza2021}, that $H^{b=0}_G(\lambda;r)$ is the coefficient of $\epsilon^r\mathscr C_\lambda$ in
\begin{equation}
\label{eq:314}
\frac 1{\mathsf{z}_\lambda}\prod_{i=1}^nG(\epsilon \mathscr J_i) ,
\end{equation}
provided that this expression is expanded as a linear combination of the elements $\mathscr C_\lambda$.

\begin{Definition}[colored monotone factorizations]
Let $\lambda\in\mathcal P $ with $|\lambda|=n$.
A \textit{monotone factorization of length $r$ and type $\lambda$} is an ordered $r$-tuple of transpositions $(\pi_1,\dots,\pi_r)$ in~$\mathfrak S_n$ such that $\pi_1\cdots\pi_r\in\mathfrak C(\lambda)$ and that, writing $\pi_i=(a_i,b_i)$ with $a_i<b_i$, we have $b_i\leq b_{i+1}$ for $i=1,\dots,r-1$.

Let $L$, $M$ be nonnegative integers.
An \textit{$(L|M)$-coloring} of a monotone factorization of length~$r$ and type $\lambda$ is a mapping $c\colon \lbrace 1,\dots,r\rbrace\to\lbrace 1,\dots,L+M\rbrace$ such that
\begin{itemize}\itemsep=0pt
\item if $1\leq i<j\leq r$ and $1\leq c(i)=c(j)\leq L$, then $b_i<b_j$;
\item if $1\leq i<j\leq r$ and $b_i=b_j$, then $c(i)\leq c(j)$.
\end{itemize}
An \textit{$(L|M)$-colored monotone factorization of length $r$ and type $\lambda$} is a monotone factorization of length $r$ and type $\lambda$ with an $(L|M)$-coloring.
\end{Definition}

\begin{Proposition}
\label{prop:or}
Let $G(z)=\frac{\prod_{i=1}^L(1+u_iz)}{\prod_{i=1}^M(1-u_{L+i}z)}$, for a set of parameters $u_1,\dots,u_{L+M}$.
For all $\lambda\in\mathcal P $ and $r\geq 0$, we have
\begin{equation*}
H_G^{b=0}(\lambda;r) = \frac {1}{|\lambda|!}\sum_{(\Pi,c)}u_{c(1)}\cdots u_{c(r)},
\end{equation*}
where the sum on $(\Pi,c)$ runs over the set of $(L|M)$-colored monotone factorizations of length~$r$ and type~$\lambda$.
\end{Proposition}

\begin{Remark}
This combinatorial model is different from the one given, e.g., in~\cite[Section~2.5]{GisonniGravaRuzza2021}.
\end{Remark}

In the interest of clarity, let us first prove a simple lemma.

\begin{Lemma}
We have
\begin{equation*}
\prod_{i=1}^n\frac{(1+\epsilon u_1 x_i)\cdots(1+\epsilon u_L x_i)}{(1-\epsilon u_{L+1}x_i)\cdots(1-\epsilon u_{L+M}x_i)} = \sum_{r\geq 0}\epsilon^r\sum_{1\leq b_1\leq\cdots\leq b_r\leq n}x_{b_1}\cdots x_{b_r} \sum_{c}u_{c(1)}\cdots u_{c(r)},
\end{equation*}
where the inner sum runs over mappings $c\colon \lbrace 1,\dots,r\rbrace\to\lbrace 1,\dots,L+M\rbrace$ satisfying the following constraints:
\begin{itemize}\itemsep=0pt
\item if $1\leq i<j\leq r$ and $1\leq c(i)=c(j)\leq L$, then $b_i<b_j$;
\item if $1\leq i<j\leq r$ and $b_i=b_j$, then $c(i)\leq c(j)$.
\end{itemize}
\end{Lemma}
\begin{proof}
Let us recall that for any finite totally ordered set $(S,\preceq)$, we have
\begin{equation*}
\prod_{s\in S}\frac 1{1-\epsilon y_s} = \sum_{r\geq 0}\epsilon^r\sum_{\substack{s_1,\dots,s_r\in S,\\s_1\preceq\cdots\preceq s_r}} y_{s_1}\cdots y_{s_r} .
\end{equation*}
Applying this identity to the set $S=\lbrace 1,\dots,n\rbrace\times\lbrace 1,\dots, L+M\rbrace$ with lexicographic order, we get
\begin{equation*}
\prod_{i=1}^n\prod_{j=1}^{L+M}\frac{1}{(1-\epsilon x_i u_j)} = \sum_{r\geq 0}\epsilon^r\sum_{1\leq b_1\leq \cdots\leq b_r\leq n} x_{b_1}\cdots x_{b_r}\sum_{\substack{c_1,\dots,c_r=1,\dots,L+M \\ c_i\leq c_j\mbox{ whenever }b_i=b_j}}u_{c_1}\cdots u_{c_r} .
\end{equation*}
Finally, we retain only terms of degree at most one in~$x_bu_c$ for $b=1,\dots, n$ and $c=1,\dots ,L$.
\end{proof}

\begin{proof}[Proof of Proposition~\ref{prop:or}]
The coefficient of $\epsilon^r$ in~\eqref{eq:314} is
\begin{align*}
&\frac 1{\mathsf{z}_\lambda}\sum_{1\leq b_1\leq \dots\leq b_r\leq n}\mathscr J_{b_1}\cdots\mathscr J_{b_r} \sum_{c}u_{c(1)}\cdots u_{c(r)}
\\
&\qquad = \frac 1{\mathsf{z}_\lambda}\sum_{1\leq b_1\leq \dots\leq b_r\leq n}\sum_{a_1=1}^{b_1-1}\cdots\sum_{a_r=1}^{b_r-1}(a_1,b_1)\cdots (a_r,b_r) \sum_{c}u_{c(1)}\cdots u_{c(r)},
\end{align*}
where the inner sum over $c$ is as in the statement of the lemma above.
In the equality, we used the definition~\eqref{eq:J} of Young--Jucys--Murphy elements.
Finally we extract the coefficient of $\mathscr C_\lambda$ in this expression by recalling that $|\mathfrak C(\lambda)|=|\lambda|!/\mathsf{z}_\lambda$.
\end{proof}

It can be shown (cf.~\cite[Remark~3]{BonzomChapuyDolega}) that an \emph{orientable} monotone Hurwitz map with profile~$\lambda$ is uniquely characterized by a monotone factorization of type~$\lambda$.
Since the notions of coloring are exactly parallel in the two settings (Hurwitz maps or factorizations), we obtain a different proof of Theorem~\ref{thm:geo} in the orientable case.

Moreover, we have recalled the group theoretical interpretation of Hurwitz numbers in the $b=0$ case, such that Theorems~\ref{thm} and~\ref{thm:Laguerre} reduce, when $\beta=2$, to the known combinatorial interpretation of large-$n$ expansions of (positive and negative) correlators of the Jacobi and Laguerre unitary ensembles in terms of monotone factorizations in the symmetric group~\cite{CundenDahlqvistOConnell,GisonniGravaRuzza2021} (although in the equivalent formulation in terms of colored monotone factorizations).

\appendix

\section{Proof of~(\ref{eq:again})}\label{appendix}

We wish to show that for any $\lambda\in\mathcal P $ with $\ell(\lambda)\leq n$, we have
\begin{equation}
\label{eq:againA}
\prod_{i=1}^{n-1}\prod_{j=i+1}^n\frac{\bigl(\tfrac\beta 2(j-i+1)\bigr)_{\lambda_i-\lambda_j}}{\bigl(\tfrac\beta 2(j-i)\bigr)_{\lambda_i-\lambda_j}} = \prod_{\square\in\mathsf{D}(\lambda)}\frac{n+\mathsf{c}_{2/\beta}(\square)}{\tfrac 2\beta\mathsf{arm}_\lambda(\square)+\mathsf{leg}_\lambda(\square)+1} .
\end{equation}
We first consider the part of the product in the left-hand side of~\eqref{eq:againA} with $i=1$, which is
\begin{equation*}
\prod_{j=2}^n\frac{\bigl(\tfrac\beta 2j\bigr)_{\lambda_1-\lambda_j}}{\bigl(\tfrac\beta 2(j-1)\bigr)_{\lambda_1-\lambda_j}} =
\frac{\bigl(\tfrac\beta 2n\bigr)_{\lambda_1-\lambda_n}}{\bigl(\tfrac\beta 2\bigr)_{\lambda_1-\lambda_2}} \prod_{j=2}^{n-1}\frac{\bigl(\tfrac\beta 2j\bigr)_{\lambda_1-\lambda_j}}{\bigl(\tfrac\beta 2j\bigr)_{\lambda_1-\lambda_{j+1}}} =
\frac{\bigl(\tfrac\beta 2n\bigr)_{\lambda_1-\lambda_n}}{\prod_{j=1}^{n-1}\bigl(\tfrac\beta 2j+\lambda_1-\lambda_j\bigr)_{\lambda_{j}-\lambda_{j+1}}} .
\end{equation*}
Let $\mathsf{D}_1(\lambda)=\lbrace(1,1),(1,2),\dots,(1,\lambda_1)\rbrace$ be the first {\it row} of the diagram of $\lambda$.
The numerator of the last expression is equal to
\begin{equation}
\label{num}
\frac{\prod_{\square\in\mathsf{D}_1(\lambda)}\bigl(\tfrac\beta 2n+\tfrac\beta 2\mathsf{c}_{2/\beta}(\square)\bigr)}
{\prod_{v=\lambda_1-\lambda_n+1}^{\lambda_1}\bigl(\tfrac\beta 2n+v-1\bigr)} ,
\end{equation}
\big(note that $\mathsf{c}_{2/\beta}(\square)=\tfrac 2\beta(j-1)$ for $\square=(1,j)\in\mathsf{D}_1(\lambda)$\big).
On the other hand, we expand the denominator
\begin{align*}
\prod_{j=1}^{n-1}\bigl(\tfrac\beta 2j+\lambda_1-\lambda_j\bigr)_{\lambda_{j}-\lambda_{j+1}} ={}&
\Biggl(\prod_{k=0}^{\lambda_1-\lambda_2-1}\bigl(\tfrac\beta 2+k\bigr)\Biggr)
\Biggl(\prod_{k=\lambda_1-\lambda_2}^{\lambda_1-\lambda_3-1}\bigl(\tfrac\beta 22+k\bigr)\Biggr)\cdots
\\
&{}\times \Biggl(\prod_{k=\lambda_{1}-\lambda_{n-1}}^{\lambda_1-\lambda_n-1}\bigl(\tfrac\beta 2(n-1)+k\bigr)\Biggr)
\end{align*}
and we quickly realize that the factors in the first product (on the right-hand side) correspond to boxes $\square\in\mathsf{D}_1(\lambda)$ with $\mathsf{leg}_\lambda(\square)=0$, those in the second one to boxes $\square\in\mathsf{D}_1(\lambda)$ with $\mathsf{leg}_\lambda(\square)=1$, and so on until the last group of factors, corresponding to boxes $\square\in\mathsf{D}_1(\lambda)$ with $\mathsf{leg}_\lambda(\square)=n-2$, such that this product equals
\begin{gather}
\frac{\prod_{\square\in\mathsf{D}_1(\lambda)}\bigl(\tfrac \beta 2 (\mathsf{leg}_\lambda(\square)+1)+\mathsf{arm}_\lambda(\square) \bigr)}{\prod_{\square\in\mathsf{D}_1(\lambda),\,\mathsf{leg}_\lambda(\square)=n-1}\bigl(\tfrac \beta 2 (\mathsf{leg}_\lambda(\square)+1)+\mathsf{arm}_\lambda(\square) \bigr)} \nonumber\\
\qquad= \frac{\prod_{\square\in\mathsf{D}_1(\lambda)}\bigl(\tfrac \beta 2 (\mathsf{leg}_\lambda(\square)+1)+\mathsf{arm}_\lambda(\square) \bigr) }{\prod_{v=\lambda_1-\lambda_n+1}^{\lambda_1}\bigl(\tfrac\beta 2n+v-1\bigr)} .\label{den}
\end{gather}
Taking the ratio of~\eqref{num} by~\eqref{den}, we get
\begin{equation*}
\prod_{j=2}^n\frac{\bigl(\tfrac\beta 2j\bigr)_{\lambda_1-\lambda_j}}{\bigl(\tfrac\beta 2(j-1)\bigr)_{\lambda_1-\lambda_j}} = \prod_{\square\in\mathsf{D}_1(\lambda)}\frac{n+\mathsf{c}_{2/\beta}(\square)}{\tfrac 2\beta\mathsf{arm}_\lambda(\square)+\mathsf{leg}_\lambda(\square)+1} ,
\end{equation*}
which implies that the validity of~\eqref{eq:againA} for the partition $(\lambda_1,\lambda_2,\dots)$ follows from the validity of the same formula for the partition $(\lambda_2,\lambda_3,\dots)$.
Therefore, the proof proceeds by induction on the length of the partition, the base case of the empty partition being trivial.\hfill$\blacksquare$

\subsection*{Acknowledgements}

I am grateful to Dan~Betea, Massimo~Gisonni, and Tamara~Grava for valuable conversations and to Valentin Bonzom, Guillaume Chapuy, and Maciej Do\l\c{e}ga for insightful and helpful correspondence.
I would also like to thank the anonymous referees for useful suggestions.
This work is supported by the FCT grant 2022.07810.CEECIND.

\pdfbookmark[1]{References}{ref}
\LastPageEnding

\end{document}